\documentclass[conference]{IEEEtran}
\makeatletter
\def\ps@headings{%
\def\@oddhead{\mbox{}\scriptsize\rightmark \hfil \thepage}%
\def\@evenhead{\scriptsize\thepage \hfil \leftmark\mbox{}}%
\def\@oddfoot{}%
\def\@evenfoot{}}
\makeatother
\usepackage{times}
\usepackage{verbatim}
\usepackage{amssymb}
\usepackage{amsmath}
\usepackage{amsthm}
\usepackage{array}
\usepackage{float}
\usepackage{graphics,graphicx,epsfig,epstopdf}
\usepackage{psfrag}
\usepackage{subfigure}
\usepackage{tikz}
\usepackage{multirow}
\usepackage{algorithm}
\usepackage{algpseudocode}

\newtheorem{thm}{Theorem}[section]
\newtheorem{cor}[thm]{Corollary}

\newtheorem{prope}[thm]{Property}

\begin{document}

\title{Beaconing-Aware Optimal Policies for Two-Hop Routing in Multi-Class Delay Tolerant Networks}
\author{\IEEEauthorblockN{Nicola~Basilico, Nicola~Gatti, Luca~Goffredi, and~Matteo~Cesana}
\IEEEauthorblockA{Dipartimento di Elettronica, Informazione e Bioingegneria, Politecnico di Milano\\
Email: \{name.surname\}@polimi.it}}

\maketitle

\begin{abstract}
In Delay Tolerant Networks (DTNs), two-hop routing compromises energy versus delay more conveniently than epidemic routing. Literature provides comprehensive results on optimal routing policies for mobile nodes with homogeneous mobility, often neglecting signaling costs. Routing policies are customarily computed by means of fluid approximation techniques, which assure solutions to be optimal only when the number of nodes is infinite, while they provide a coarse approximation otherwise. This work addresses heterogeneous mobility patterns and multiple wireless transmission technologies; moreover, we explicitly consider the beaconing/signaling costs to support routing and the possibility for nodes to discard packets after a local time. We theoretically characterize the optimal policies by deriving their formal properties. Such analysis is leveraged to define two algorithmic approaches which allow to trade off optimality with computational efficiency. Theoretical bounds on the approximation guarantees of the proposed algorithms are derived. We then experimentally evaluated them in realistic scenarios of multi-class DTNs.
\end{abstract}



\section{Introduction}
Delay--tolerant networks (DTNs) are sparse and/or highly mobile wireless ad hoc networks  which assure no continuous connectivity. Examples of such networks are those operating in mobile or hash terrestrial environments, or interplanetary networks. Disruption may occur because of the limits of wireless radio range, sparsity of mobile nodes, energy resources, attacks, and noise. One central problem in DTNs is the routing of packets from a source towards the desired destinations. When no \emph{a priori} information is available  over the mobility pattern of the nodes, a common technique for overcoming lack of connectivity is to disseminate multiple copies of the packet in the network: this enhances the probability that at least one of them will reach the destination node within a given temporal deadline. This is referred to as \emph{epidemic--style} forwarding, because, alike the spread of infectious diseases, each time a packet--carrying node encounters a new node not having a copy thereof, the carrier may infect this new node by passing on a packet copy; newly infected nodes, in turn, may behave similarly. The destination receives the packet when it meets an infected node. However, epidemic routing is very energy consuming and a convenient compromise of energy versus delay compared to epidemic routing is provided by two-hop routing~\cite{DBLP:journals/pe/GroeneveltNK05} where the infection is limited at the contacts between the source node and intermediary nodes, that is, the source node passes on the packets to be transmitted to all the mobile node she encounters (provided that these last ones do not already have a copy of the packet in their local buffer), and the "infected" mobile node can deliver the packets they are carrying only to the final destination. 

We target here the design of optimal two-hop routing policies for DTNs as in~\cite{DBLP:journals/ton/AltmanABP13,DBLP:journals/tac/AltmanBP11,DBLP:conf/wiopt/PellegriniAB10}. In two-hop routing, a routing policy controls the decisions taken by the source node to forward/not to forward a given packet to a given mobile node she is encountering at a given time instant.  In this work, we further consider network scenarios where mobile nodes are categorized in district multiple classes which may capture different mobility patterns and different available communication technologies onboard and we study the "shape" of optimal routing policies under heterogeneous mix of mobile nodes classes.  

In the seminal work~\cite{DBLP:conf/infocom/AltmanNPM09}, the authors study optimal static and dynamic control (proved to be threshold based) policies for two-hop DTN when mobile nodes are homogeneous. In this case, optimal policies can be found in closed form. Furthermore, the authors show that when the parameters are unknown it is still possible to obtain a policy that converges to the optimal one by using some adaptive auto--tuning mechanism. Extensions of such adaptive mechanism are proposed in~\cite{DBLP:journals/comcom/GuerrieriCPMM10}. Heterogeneous scenarios are studied in~\cite{DBLP:conf/wiopt/PellegriniAB10} where mobile nodes belong to multiple distinct classes and the routing policy may be class-dependent; the authors resort to fluid approximation to characterize the optimal routing policies; however, even if the proposed approach performs well when the number of mobile nodes goes to infinity, fluid approximation may lead to coarse solutions for finite numbers of mobile nodes (e.g., in other fields in~\cite{DBLP:journals/ior/PerryW11, DBLP:journals/orl/Randhawa13}, fluid approximation is shown to perform well when the number of users is over 200). Optimal control techniques are further proposed under the assumption that the number of copies of the packet is monotonically increasing in time. Along the same lines, in~\cite{DBLP:conf/wd/ChahinAPA11}, the authors provide the closed form structure of the controller and provide stochastic algorithms for the learning of the parameters. Timers have been proposed to be associated with messages when stored at mobile nodes, so that after a given threshold the message is discarded. Their optimization under fluid approximation is discussed in~\cite{DBLP:journals/tac/AltmanBP11}. In~\cite{DBLP:journals/tsmc/TembineAAH10} and in~\cite{DBLP:journals/ton/BanerjeeCL10}, the authors optimize network performance by designing the mobility of message relays. Evolutionary game theory has been adopted in~\cite{DBLP:journals/cn/AzouziPSK13} to incentivize mobile nodes.  Many other related problems have been studied. In~\cite{DBLP:journals/twc/AltmanSP13} the authors deal with the problem of transferring large files from source to destination when all packets are not available at the source prior to the first transmission. In~\cite{DBLP:journals/ton/AltmanP11}, a class of replication mechanisms for packets in the network that include coding in order to improve the probability of successful delivery within a given time limit. In~\cite{el2010evolutionary}, nodes are considered selfish/non--cooperative and the problem is to incentivize the mobile nodes to forward the packet. The first node that delivers the packet receives a unitary reward. The only cost of the mobile nodes is due to beaconing. The authors study the best strategies by using evolutionary game theory tools.  In~\cite{chen2010mobicent}, the authors focus on multi--hop routing scenarios in the attempt to design incentives for the mobile nodes to avoid that these nodes behave strategically in terms of edge hiding and edge insertion. The proposed model works only when beaconing costs are not present and each mobile node is subject to the same cost. A similar technological, but not technical, approach is proposed in~\cite{zhu2009smart}. In~\cite{jain2004routing}, the authors study different routing strategies when the nodes have a limited buffer and provide an experimental comparison of the strategies. A similar work is presented in~\cite{li2010routing} except that here the authors introduce social information about the nodes. Cooperative aspects are treated in~\cite{ning2011incentive} and~\cite{niyato2010coalition}.

Differently than the aforementioned reference literature, (i) we explicitly include in the definition of the optimal routing policy the cost for beaconing messages exchanged by mobile nodes to support packet forwarding; (ii) we allow mobile nodes to discard the packets they are carrying upon expiration of a local temporal deadline (the number of copies of a packet is thus no-longer monotonically as customarily assumed in the reference literature); (iii) we do not resort to fluid approximation, but rather propose optimization algorithms for finite mobile node populations; in details, we introduce an algorithm to find optimal routing policies running in exponential time, as well as approximate polynomial-time algorithms when the number of nodes is finite; in the latter case, formal  approximation theoretical bounds are also derived. Finally, we provide a thorough experimental evaluation with realistic settings of the proposed algorithms in terms of approximation ratio, scalability in the number of classes, further evaluating the impact of network parameters onto the optimal routing policies.

\section{Problem statement}

\subsection{Scenario}
We consider an environment populated by one {\em source} node, one {\em sink} node and multiple {\em mobile nodes}. \emph{Sink} and \emph{source} node may as well be mobile. A packet is initially held by the source and it must be delivered to the sink no later than time $\tau$. We consider a discrete time representation developing in slots of fixed duration $\Delta$ and we denote the total number of useful time slots as $K = \lfloor{\tau/\Delta}\rfloor$, where the $k$--th time slot corresponds to the time interval $[k \Delta, (k+1)\Delta)$. Mobile nodes are divided according to a set of classes $C$. Each class encodes the features of its nodes, including their {\em mobility profile} and {\em transmission technology}. In particular, given a class $c \in C$, $N_c$ indicates the total number of mobile nodes belonging to that class; $t_c$ indicates the {\em local} time to live, namely the amount of time for which each mobile node of class $c$ will keep a local copy of the packet before discarding it and not accepting it again in the future. The mobility profile of a node of class $c$ is given by the average speed $v_c$. Finally, we describe the transmission technology for class $c$ with the following parameters:
\begin{itemize}
\item the communication range for class $c$ is denoted by  $R_c$;
\item the \emph{beaconing cost}  $\beta_c$ captures the energy consumed  for the connection control and signaling procedures to support packet transmissions in class $c$,; as an example, the beaconing cost may capture the energy consumed to send/receive beaconing messages to discover nearby mobile nodes;   
\item the \emph{transmission cost}  $\rho_c$ is the energy consumed to transmit a packet to a recipient in  class $c$.
\end{itemize}
Different classes $c$ can have the same technology, e.g., WiFi, sharing beaconing costs. We denote by $\omega\in \Omega$ a technology, and we denote by $C_{\omega}\subseteq C$ the subset of classes adopting technology $\omega$. With overriding of notation, we use indistinguishably $\beta_\omega$ and $\beta_c$ when $c\in C_\omega$.
Table~\ref{tab:class_parameters} summarizes all the parameters that define a generic class $c \in C$.

\begin{table}
\caption{Set of parameters for a generic class $c \in C$}
\centering
\begin{tabular}{|c|c|}
\hline
  $N_c$ & number of nodes in class $c$ \\
\hline
  $t_c$ & packet's local time to live \\
\hline
  $v_c$ & average speed \\
\hline
  $R_c$ & communication range \\
\hline
  $\beta_c$ & beaconing cost \\
\hline
  $\rho_c$ & transmission cost \\
\hline
\end{tabular}
\label{tab:class_parameters}
\vspace{-0.4cm}
\end{table}

Transmission opportunities between two nodes are given by contacts taking place when each node is within the communication range of the other one. We restrict our setting to a 2--hop routing scheme, where each mobile node can receive the packet only from the source and can forward it only to the sink. Contacts at the source and at the sink are assumed to follow a multi--class Poisson law, where the arrival rate for nodes of class $c$ (either at the source or the sink) is denoted by $\lambda_c$ and computed according to~\cite{DBLP:conf/infocom/AltmanNPM09}:
$$\lambda_c=\frac{8wR_cv_c}{\pi L^{2}}$$
where $w$ is a constant set to $1.3693$ and $L$ is the radius of circle in which the nodes move.

When a contact is made between the source and a mobile node that did not receive the packet, the source forwards it according to a {\em forwarding policy}  $\boldsymbol \mu$ which depends on the current time and the recipient's class. Given a time slot $k$ and a class $c$, the policy profile at time $k$ is $\boldsymbol \mu(k) = (\mu_{1}(k),\ldots,\mu_{|C|}(k))$ where $\mu_c(k)$ indicates the forward probability at time slot $k$ for class $c$; we also denote with $\mu_c$ the entire policy for such class $c$. In general, when the packet is forwarded, some energy is spent and the packet's delivery probability is increased. We denote with $F_D(\boldsymbol \mu, K)$ the probability of delivering the message before time $K \Delta$ given policy profile $\boldsymbol \mu$. Obviously, such value is prevented from growing indefinitely by a budget constraint. We call $\Psi$ the upper bound on the total spent energy (including both beaconing and transmissions).

\subsection{Formal model}

We define $X_{c,k}(\boldsymbol \mu)$ as the random variable expressing the number of mobile nodes of class $c$ that have received the packet by time slot $k$, while $Y_{c,k}(\boldsymbol \mu)$ is a random variable expressing the number of mobile nodes of class $c$ that still keep a copy of the packet at time slot $k$. These variables both depend on $\boldsymbol\mu$ and are, in general, different. Indeed, since a mobile node can both receive and discard a packet before time slot $k$, we have that $Y_{c,k}(\boldsymbol \mu) \le X_{c,k}(\boldsymbol \mu)$. Furthermore, we denote by  $Q_{c,k,k'}(\mu_c)$ the probability that a mobile node of class $c$ does not receive any packet in time slots $k,\ldots,k'$ as function of $\mu_c$. The expected number of mobile nodes of class $c$ that receive a packet in time slots $0,\ldots,k$ is:
\[
\mathbb{E}[X_{c,k}(\boldsymbol \mu)] = N_c \cdot (1- Q_{c,0,k}(\mu_c))
\]
where $Q_{c,k,k'}(\mu_c) = e^{-\lambda_c \Delta \sum_{i=k}^{k'}\mu_c(i)}$.
The expected number of mobile nodes of class $c$ that have the packet at time slot $k$ is:
\[
\mathbb{E}[Y_{c,k}(\mu_c)] = N_c \cdot (1- Q_{c,\max\{0,k-t_c\},k}(\mu_c))
\]
The probability that a packet is delivered by $k\Delta$ to the sink is:
\begin{equation}
F_D(\boldsymbol \mu,k) = 1 - \prod_{c\in C}\prod_{h=0}^{k-1} X_{c,h}^*(\lambda_c\Delta,\boldsymbol \mu)\label{eq:OF}
\end{equation}
where
\[
X_{c,h}^*(s,\boldsymbol \mu) = \mathbb{E}[e^{-sY_{c,h}(\boldsymbol \mu)}]
\]
The budget constraint is formulated as:
\begin{multline}
\sum_{c\in C}  \rho_c N_c (1- Q_{c,0,K}(\mu_c))+ \\\sum_{\omega\in \Omega}\sum_{k=0}^{K-1}\beta_\omega\cdot \left( 1 - \prod_{c\in C_\omega}\Big(1-\mu_c(k)\Big)\right)  \leq \Psi \label{equation:budget}
\end{multline}

The left term of the inequality adds up the expected transmission costs with the expected beaconing costs for class $c$, given a policy profile $\boldsymbol\mu$. In particular, transmission costs are obtained by multiplying $\rho_c$ by the expected number of nodes that will receive the packet from slot 0 to slot $K$; on the other side, a beaconing costs of $\beta_c$ is paid for each time slot $h$ with a probability of $\mu_c(h)$, i.e., when a packet transmission could be made.

We aim at finding the optimal policy $\boldsymbol \mu^*$ that maximizes $F_D(\boldsymbol \mu,K)$ without violating budget constraint (\ref{equation:budget}).

\subsection{Problem properties}
We now show some  theoretical properties that we shall exploit in addressing the optimization problem introduced above.

\begin{prope}\label{prop:saturate_or_saturate}
Optimal policies either completely consume the budget or prescribe that all the classes transmit for all the slots.\label{property:budget}\vspace{-0.3cm}
\end{prope}
\begin{proof}
It is easy to see that $F_D(\boldsymbol \mu,K)$ is monotonically increasing in $\sum_{h=0}^{K-1} \mu_c(h)$ and that, as a consequence, transmitting for a larger (expected) number of slots cannot result in a lower delivery probability.
\end{proof}

Similarly to what proposed in~\cite{DBLP:conf/infocom/AltmanNPM09}, we define a threshold--based policy $\mu_c$ as:\vspace{-0.6cm}
\[
\mu_c(k) = 
\begin{cases}
1		&	k<h_c	\\
\alpha	&	k=h_c	\\
0		&	k>h_c
\end{cases}
\]
where $\alpha\in [0,1)$. As in the single--class case, optimal policies are threshold based.
\begin{prope}\label{prope:threshold_optimal}
Optimal policies are threshold based.\vspace{-0.3cm}
\end{prope}
\begin{proof}
The delivery delay c.d.f. is $F_D(\boldsymbol\mu,K) = 1 - \prod_{c}\Gamma_c(s)$, where $\Gamma_c(s) = \prod_{h=0}^{K-1} X_{c,h}^*(s)$ and $s = \lambda_c\Delta$. Let us denote with $\mu_c$ a \emph{non--threshold} policy for class $c$ and with $\hat{\mu}_c$ a policy obtained by shifting to the left all the non-empty slots of $\mu_c$ and by rounding them so that $\hat{\mu}_c$ matches the definition of threshold policy introduced above. For any $(\mu_c, \hat{\mu}_c)$ obtained in this way we have that $\Gamma_c(s) \ge \hat{\Gamma}_c(s)$ and therefore 
\[1 - \Gamma_c(s)\cdot \Gamma_{-c}(s) \le 1-\hat{\Gamma}_c(s)\cdot \Gamma_{-c}(s)\]
that is, for any given joint policy $\boldsymbol \mu$, if we substitute the marginal policy of a class $c$ with its threshold version the probability of delivery within $K$ time steps will not decrease.
\end{proof}

\begin{prope}\label{prop:all_fractional}
Optimal policies can prescribe non--integer thresholds for all the classes.\vspace{-0.3cm}
\end{prope}
\begin{proof}
Consider, for instance, a two--class instance with: $K=20$, $\Delta=100$, $\Psi=0.7$, $N_1=1$, $N_2=2$, $\lambda_1=21\times10^{-5}$, $\lambda_2=20\times10^{-5}, t_1=t_2=K$. We approximated the optimal policy profile by discretizing the values of $h_c$ with a fine grid with step $0.01$. In addition to these points, we considered all the points in which the threshold of one class is integer and the threshold of the other class is calculated in such a way the budget is completely consumed. We evaluated the objective function at all these points and select the maximum. The approximately optimal policy is $h_1=7.87$, $h_2=15.99$. Hence, at the optimum, a fractional part is assigned to both classes.
\end{proof}

\begin{prope}
The optimization problem with objective function Eq. \ref{eq:OF} and constraint Eq. \ref{equation:budget} is nonlinear and nonconvex.\vspace{-0.3cm}
\end{prope}
\begin{proof} Nonlinearity is trivial. Nonconvexity is proved by showing the nonconvexity of the feasibility region by computing the Hessian matrix of the budget constraint  (\ref{equation:budget}) to which we will refer here with $u$ (notice that we restrict our attention on threshold policies from now). Hessian matrix $\mathcal{H}$ is:
\normalsize{
\begin{multline*} 
\mathcal{H}=\left[
\begin{array}{ccc}
\frac{\partial^2 u}{\partial h_1^2} 	& 		 		& 				0			\\
							&		\ddots	&							\\
 0							& 				&	\frac{\partial^2 u}{\partial h_{|C|}^2}	
\end{array}
\right] 
=\\ 
\left[
\begin{array}{ccc}
-N_1\lambda_1^2\Delta^2 e^{-\lambda_1\Delta h_1} 	& 		 		& 				0			\\
							&		\ddots	&							\\
 0							& 				&	-N_{|C|}\lambda_{|C|}^2\Delta^2 e^{-\lambda_{|C|}\Delta h_{|C|}} 
\end{array}
\right] 
\end{multline*} 
}
\normalsize{
It can be easily seen that all the eigenvalues are strictly negative for every policy profile $\boldsymbol \mu$ and therefore the feasibility region is non convex.}\end{proof}

The above properties show that the optimization problem is hard. In particular, the adoption of non--convex programming techniques required by the nature of the problem cannot assure to find of global optimal solutions. For these reasons, we focus on the problem of developing approximation algorithms and of studying their theoretical and empirical approximation errors.

\section{Approximation algorithms}

In this section, we introduce two  approaches to compute the optimal (threshold) policy. We provide the formalization of two approximation algorithms and we discuss their guarantees on the solution quality loss.

\subsection{Non--polynomial--time approximation scheme}
\label{subsection:gridalgorithm}

We start by defining an approximation scheme that does not run in polynomial time, but for which optimality losses can be arbitrarily bounded. We overconstrain the optimization problem, allowing only a single class to have a fractional threshold in its associated policy. This additional constraint is likely to introduce worsenings in the solution quality (see Property~\ref{prop:all_fractional}) but it allows us to provide a combinatorial version of the optimization problem. Indeed, once all the classes except one have been assigned integer policies, the potentially fractional policy of the remaining class is univocally determined either by the policy that consumes all the remaining budget or the one that transmits until the last useful time slot (see Property~\ref{prop:saturate_or_saturate}). In addition, we split each slot of length $\Delta$ in sub--slots of length $\epsilon\Delta$ where, for simplicity, $\frac{1}{\epsilon}\in \mathbb{N}$. 

This overconstrained problem can be solved optimally by using an  enumeration algorithm: we enumerate all the feasible threshold policies and we select the best one (see Property~\ref{prope:threshold_optimal}). We report in Algorithm~\ref{alg:gridepsilon} the necessary steps. At Step~1, the algorithm initializes $F^*$ to be zero. If it is not possible to entirely consume the budget, then the optimal policy profile is to assign $h_c = (K-1)/\epsilon$ to each class $c$ (Step~2--3). Otherwise, the algorithm enumerates all the classes $c$, and for each class $c$ it enumerates all the policy profiles  $\boldsymbol\mu = (\mu_{c},\boldsymbol\mu_{-c})$ s.t. $\boldsymbol\mu_{-c}$ is integer  and budget $\Psi$ is entirely consumed (Step~5--6). Finally, we keep trace of the best policy found so far.

\begin{algorithm}\caption{$\epsilon$--grid search}\label{alg:gridepsilon}
\begin{algorithmic}[1]
\State $F^* \gets 0$
\If{$\boldsymbol \mu$ s.t. $h_c = \frac{K-1}{\epsilon}$ for all $c$ is feasible}

\State $\boldsymbol \mu^*=\boldsymbol \mu$

\Else

\For{$c \in C$} 

\For{\text{every} $\boldsymbol\mu = (\mu_{c},\boldsymbol\mu_{-c})$ s.t. $\boldsymbol\mu_{-c}$ is integer  and budget $\Psi$ is entirely consumed}

\If{$F_D(\frac{K-1}{\epsilon},\boldsymbol\mu) > F^*$} 

\State $\boldsymbol\mu^* \gets \boldsymbol\mu$

\State $F^* \gets F_D(\frac{K-1}{\epsilon},\boldsymbol\mu)$

\EndIf 

\EndFor

\EndFor

\EndIf
\end{algorithmic}
\end{algorithm}

We now describe an efficient method to enumerate all and only the feasible policy profiles $\boldsymbol\mu = (\mu_{c},\boldsymbol\mu_{-c})$ s.t. $\boldsymbol\mu_{-c}$ is integer  and budget $\Psi$ is entirely consumed. First, we build a lexicographic order over $C_{-c}= C\setminus c$ and we scan lexicographically the classes in $C_{-c}$. Then, for each $c'\in C_{-c}$ we determine the range of feasible values for $h_{c'}$ on the basis of the policies assigned to the classes $c''\succ c'$ as follows: 
\begin{multline*}
I_{c'}(\boldsymbol\mu_{-c'})=\\\left\{ \max\left\{0, \lceil r_{c'}(\overline{\boldsymbol \mu}_{-c'})\rceil\right\}, \ldots, \min\left\{\frac{K-1}{\epsilon}, \lfloor r_{c'}(\underline{\boldsymbol \mu}_{-c'})\rfloor\right\} \right\}
\end{multline*}
where $r_{c'}(\boldsymbol\mu_{-c'})$ is computed as follows:
\begin{itemize}
\item initially we compute:
\begin{multline*}
r_{c'}(\boldsymbol\mu_{-c'})=-\frac{\log\left(\frac{N_{c'}+A-\frac{\Psi}{\rho}+\underset{c'':c''\neq c'}{\max}\{\frac{\beta}{\rho} h_{c''}\}}{N_{c'}}\right)}{\lambda_{c'}\Delta\epsilon}
\end{multline*}
where
\begin{multline*}
A=\sum\limits_{c'':c''\neq c'}N_{c''}\left(1-e^{-\lambda_{c''}\Delta\epsilon \sum\limits_{k=0}^{\frac{K-1}{\epsilon}}\mu_{c''}(k)}\right)
\end{multline*}
\item if $r_{c'}(\boldsymbol\mu_{-c'})>\underset{c'':c''\neq c'}{\max}\{\frac{\beta\epsilon}{\rho} h_{c''}\}$, then $r_{c'}(\boldsymbol\mu_{-c'})$ is the solution (that can be approximated by means of Netwon algorithm) of the following equation:
\begin{multline*}
N_{c'}(1-e^{-\lambda_c\epsilon\Delta r_{c'}(\boldsymbol\mu_{-c'})})\\+\sum\limits_{c'':c''\neq c'}N_{c''}\left(1-e^{-\lambda_{c''}\Delta\epsilon \sum\limits_{k=0}^{\frac{K-1}{\epsilon}}\mu_{c''}(k)}\right)\\-\frac{\Psi}{\rho}+\frac{\beta\epsilon}{\rho}r_{c'}(\boldsymbol\mu_{-c'})=0
\end{multline*}
\end{itemize}
and where $\overline{\boldsymbol \mu}_{-c'}$ and $\qquad\underline{\boldsymbol \mu}_{-c'}$ are defined in the following way:
\[
\overline{\boldsymbol \mu}_{-c'}=
\begin{cases} 
\overline{\mu}_{c''} = \mu_{c''} & c'' \succ c' \\
\overline{\mu}_{c''} : h_{c''}=\frac{K}{\epsilon}-1 & c' \succ c'' \\
\overline{\mu}_{c} : h_{c}=\frac{K}{\epsilon}-1
\end{cases}
\]
\[
\qquad\underline{\boldsymbol \mu}_{-c'}=
\begin{cases} 
\underline{\mu}_{c''} = \mu_{c''} & c'' \succ c' \\
\underline{\mu}_{c''} : h_{c''}=0 & c' \succ c'' \\
\underline{\mu}_{c} : h_{c}=0 
\end{cases}
\]

Once the previous steps are done, for every element in $R_{c'}(\boldsymbol\mu_{-c'})$, we assign it to $h_{c'}$ and go to the next class according to the lexicographic order. Finally, once the policies of all the classes $c'\in C_{-c}$ have been assigned, the policy of $c$ is easily given by $h_c=r_c(\boldsymbol \mu_{-c})$.

\begin{thm}
The above method enumerates all and only the feasible policies consuming exactly the budget in which at most one $h_c$ is fractional.\vspace{-0.3cm}
\end{thm}
\begin{proof}
We need to prove that: 
\begin{itemize}
\item all the policies except $\mu_c$ are integer, 
\item the budget is exactly consumed, and 
\item all and only the feasible policies are enumerated 
\end{itemize}
The first two points are trivial by construction  (given that the policy of $c$ is the only potentially non--integer and is computed as the policy that consumes the budget given the policies of all the other classes). To prove the third point, we observe that $I$ is always a well--defined range. Indeed, $\lfloor r_{c'}(\underline{\boldsymbol \mu}_{-c'})\rfloor$ returns the largest $h_{c'}$ that consumes exactly the remaining budget given the budget consumed by all the classes preceding $c'$ in the lexicographic order. Assigning a policy larger than $\min\left\{\frac{K-1}{\epsilon}, \lfloor r_{c'}(\underline{\boldsymbol \mu}_{-c'})\rfloor\right\} $ does not allow one to consume entirely the budget. If the policies assigned to the previous classes are feasible, then $\lfloor r_{c'}(\underline{\boldsymbol \mu}_{-c'})\rfloor$ is always non--negative. As well, $\lfloor r_{c'}(\underline{\boldsymbol \mu}_{-c'})\rfloor$ returns the smallest $h_{c'}$ that consumes exactly the remaining budget given the budget consumed by all the classes preceding $c'$ in the lexicographic order and assuming that the classes that succeed transmit all the slots. Assigning a policy smaller than $\max\left\{0, \lceil r_{c'}(\overline{\boldsymbol \mu}_{-c'})\rceil\right\}$ does no allow one to consume entirely the budget. Even in this case, if the policies assigned to the classes that precede $c$ is feasible, then  $\lfloor r_{c'}(\underline{\boldsymbol \mu}_{-c'})\rfloor$ is always smaller than $\frac{K-1}{\epsilon}$. Thus, by construction, for each policy assigned to class $c'$ belonging to $I$, it is always possible to find a feasible policy for the succeeding classes.\end{proof}

The number of policies enumerated by Algorithm~\ref{alg:gridepsilon} is exponential in the number of classes, being $O((\frac{K-1}{\epsilon})^{|C|-1})$.

We can derive a theoretical lower bound over the quality of the solution found by Algorithm~\ref{alg:gridepsilon} w.r.t. the optimal solution of the optimization problem.
\begin{thm}
\label{thm:lower_bound}
Called $\tilde{F}_D$ the value of the solution returned by Alorithm~\ref{alg:gridepsilon} and $F^*_D$ the value of the optimal solution, we have $\dfrac{\tilde{F}_D}{F^*_D}\geq \dfrac{1-(\frac{1}{2})^{\frac{K-1}{\epsilon}}}{1-(\frac{1}{2})^{|C|\frac{K-1}{\epsilon}}}$.\vspace{-0.3cm}
\end{thm}
\begin{proof}
Call $\boldsymbol \mu^*$ the optimal policy profile and call $\tilde{\boldsymbol \mu}_c$ the policy profile in which $\tilde{h}_{c'}=\lfloor h_{c'}^*\rfloor$ for all $c'\neq c$ and $\tilde{h}_{c}=h_{c}^*$. 
Obviously, $F^*_D\geq F_D(K,\tilde{\boldsymbol \mu}_c)$. In addition, it is obvious that $F^*_D\geq \tilde{F}_D \geq \max_c\{F_D(K,\tilde{\boldsymbol \mu}_c)\}$. This is because $\tilde{\boldsymbol \mu}_c$ is a feasible policy profile  in which at most one policy is fractional that is not assured to consume exactly the budget. We can write a lower bound to $F_D\left(\frac{K}{\epsilon},\tilde{\boldsymbol \mu}_c\right)$ as:  
\begin{multline}
F_D\left(\frac{K}{\epsilon},\tilde{\boldsymbol \mu}_c\right) = 1 - \prod_{c'\in C}\prod_{k=0}^{\frac{K-1}{\epsilon}} X_{c',k}^*(\lambda_{c'}\epsilon\Delta,\tilde{\boldsymbol \mu}_c) \geq \\1 - \prod_{k=0}^{\frac{K-1}{\epsilon}} X_{c,k}^*(\lambda_{c}\epsilon\Delta,\tilde{\boldsymbol \mu}_c)
\end{multline}
By using such lower bound over $F_D\left(\frac{K}{\epsilon},\tilde{\boldsymbol \mu}_c\right)$, we can write:
\[
\dfrac{\tilde{F}_D}{F^*_D}\geq \max_c\left\{\dfrac{1 - \prod_{k=0}^{\frac{K-1}{\epsilon}} X_{c,k}^*(\lambda_{c}\epsilon\Delta,\boldsymbol \mu^*)}{1 - \prod_{c'\in C}\prod_{k=0}^{\frac{K-1}{\epsilon}} X_{c',k}^*(\lambda_{c'}\epsilon\Delta,\boldsymbol \mu^*)}\right\}
\]
since, given $\tilde{\boldsymbol \mu}_c$ and $\boldsymbol \mu^*$, we have $\tilde{h}_c=h^*_c$. Thus, we are interested in:
\[
\min\max_c\left\{\dfrac{1 - \prod_{k=0}^{\frac{K-1}{\epsilon}} X_{c,k}^*(\lambda_{c}\epsilon\Delta,\boldsymbol \mu^*)}{1 - \prod_{c'\in C}\prod_{k=0}^{\frac{K-1}{\epsilon}} X_{c',k}^*(\lambda_{c'}\epsilon\Delta,\boldsymbol \mu^*)}\right\}
\]
where the minimization is over all the parameters. Although the definition of $X^*$ is intricate, a bound can be derived disregarding the exponential nature of all the $X^*$ and considering them as arbitrary values in $[0,1]$. In this case, for reasons of symmetry, the values that minimize the maximum ratio prescribe $X_{c,k}^*=\frac{1}{2}$ for all $c$. This leads to the bound stated in the theorem.
\end{proof}
Notice that the theoretical lower bound does not depend on whether the beaconing costs are present. The worst case is when $K=1$ and $|C|\rightarrow \infty$, obtaining a ratio of $1-\frac{1}{2}^\frac{1}{\epsilon}$. However, it can be observed that the  worst case ratio goes to one exponentially in $\frac{1}{\epsilon}$. Thus we can obtain a good approximation ratio with a small value of $\frac{1}{\epsilon}$, e.g, the theoretical lower bound over the approximation ratio is about $1-10^{-4}$ when $\frac{1}{\epsilon} =10$. Algorithm~\ref{alg:gridepsilon} is an approximation scheme (AS), given that the ratio goes to one as $\epsilon$ goes to zero. However, it is not a fully polynomial time AS (FPTAS), its complexity not being polynomial in all the parameters.

\subsection{Polynomial--time approximation algorithm}
In this section, we discuss a  heuristic approach to approximate the optimal policy in polynomial time. We start by providing Algorithm~\ref{alg:greedy}, a  method that greedily maximizes an objective function $G_i$.  We shall consider two versions of this function, denoted with $G_1$ and $G_2$, and we will discuss approximation bounds guaranteed by their employment. 

\begin{algorithm}
\caption{$i$ - greedy construction}\label{alg:greedy}
\begin{algorithmic}[1]
\State $\mu_1(k), \ldots, \mu_{C}(k) \gets 0 \; \forall k$
\State $k_1, \ldots, k_{C} \gets 0$
\State $F^* \gets 0$
\While{Constraint~(\ref{equation:budget}) is satisfied} 
\For{every class $c$}
    \State $\mu_c(h_c + 1) \gets \min\{1,r_c(\boldsymbol\mu)\}$
    \State$\delta_c \gets G_i(\boldsymbol\mu) $
    \State $\mu_c(k_c + 1) \gets 0$
\EndFor
\State $c^* \gets \displaystyle\arg\max_{c \in C}\{\delta_c\}$
\State $\mu_{c^*}(k_{c^*} + 1) \gets \min\{1,\hat{\mu}_{c^*}\}$
\State $k_{c^*} \gets k_{c^*} + 1$
\State $F^* \gets F^* + \delta_{c^*}$
\EndWhile
\end{algorithmic}
\end{algorithm}

Algorithm~\ref{alg:greedy} works on the same discrete-time representation we introduced above, where each time slot has a temporal length of $\epsilon\Delta$. It starts from an initial empty policy and it proceeds considering only integer threshold policies.  At each iteration, it appends a locally optimal time slot for a class $c$ meaning that such class will transmit with probability of $1$ for an additional subsequent time slot. Class $c$ is selected as the one that would introduce the largest gain in $G_i$ if a slot is assigned to it.  We denote with $k_c$ the integer index for class $c$, referring to the last allocated time slot. Similarly, $\delta_c$ denotes the discrete marginal gain of $G_i$ obtained by allocating a slot to class $c$ in the current policy.

\subsubsection*{First version, locally optimizing $F_D$} 
In the first version of Algorithm~\ref{alg:greedy}, we consider the maximization of the marginal gain of $F_D$, i.e., the delivery probability. Namely, 
at step~(7) it holds that $G_1(\boldsymbol\mu) = F_D(K,\boldsymbol\mu) - F^*$.
In this case, $\delta_c$ represents the benefit, in terms of delivery probability, that an additional time slot for class $c$ would introduce at the current iteration. By exploiting a result presented in~\cite{krause12survey} we are able to provide a  bound on the solution quality obtained with this version of the greedy algorithm. The result we shall use can be summarized as follows (see~\cite{krause12survey} for details). 
\begin{thm}[From~\cite{krause12survey}]\label{thm:submodular_bound}
Given a ground set $\Omega$, a set function $F : 2^{\Omega} \rightarrow \mathbb{R}$, and a positive integer $W \in \mathbb{N^+}$, let us consider the problem finding $\hat{S}^* = \arg\max_{S \subseteq \Omega, |S| \le W}F(S)$. Then if $F$ is submodular, we have that for every integer $0 \le l \le W$
$$F(S_l) \ge  (1 - e^{-l/W}) F(\hat{S}^*)$$
where $S_l \in \Omega$ is the set built after $l$ iterations of the following greedy element--selection rule
\begin{equation}
S_i = 
\begin{cases}
\emptyset		&	\text{if}\;$i=0$	\\
S_{i-1} \cup \arg\max_{s \in \Omega}F({S_{i-1} \cup \{s\})}		&	\text{else}
\end{cases}
\label{eq:greedy-selection-rule}
\end{equation}
 \end{thm}
Theorem~\ref{thm:submodular_bound} states that greedily maximizing a submodular set function introduces a bounded suboptimality. Eventually, the bound converges to $(1-\frac{1}{e})$ ($\approx 0.63$) when $l = W$, that is, when the maximum number of selections allowed by the cardinality constraint is made.

In order to apply this result to Algorithm~\ref{alg:greedy}, we need to show that the problem of finding an optimal integer policy can be expressed as the maximization of a submodular set function subject to a cardinality constraint. This similarity can be shown by using the following simple formalism. Let us assume that each element in the ground set $e \in \Omega$ is a pair $(c,k)$ where $c \in C$ and $k \in \{1 \ldots \frac{K-1}{\epsilon}\}$. Then, every subset $S \subseteq \Omega$ can be uniquely associated to an integer policy that we denote as $\mu^S$. Intuitively, the correspondence between $S$ and $\mu^S$ is obtained by the following construction rule:
\[
\mu^S_c(k) = 
\begin{cases}
1		&	(c,k) \in S	\\
0		&	\text{else}
\end{cases}
\]
Therefore, the objective function for a policy $\mu^S$ can be rewritten as a set function $F(K, S)$.

The second necessary step is to derive a cardinality constraint to define the problem's feasibility region. In our problem, the feasibility of a policy is determined by the budget limit, namely by Constraint~(\ref{equation:budget}). For this reason, ideally one would like to find a $W$ such that $|S| > W$ if and only if $\mu^S$ violates Constraint~(\ref{equation:budget}). However, it can be easily shown that budget feasibility cannot be expressed with a cardinality constraint. The reason is straightforward. The budget of a policy does not solely depend on the number of transmitting slots, but also on how those slots are distributed among the different classes. Nevertheless, a necessary (not sufficient) cardinality upper bound can be determined via the following theorem. 
\begin{thm}\label{thm:cardinality_constraint}
Any feasible threshold integer policy cannot assign full probability of transmission to more than $W=\min \{ \max_c\{r_c(\boldsymbol\mu^{\emptyset})\}, \frac{K-1}{\epsilon}\}$, where $\boldsymbol\mu^{\emptyset}$  is the empy policy.\vspace{-0.3cm}
\end{thm}
\begin{proof}
Let us assume that $\hat{c}=\arg\max_c\{r_c(\boldsymbol\mu^{\emptyset})\}$. Then, consider a threshold policy $\mu^S$ where $|S| > W$. If $\mu^S$ is budget--feasible then, by definition, the policy obtained in this way should be feasible too: for every $(c,k) \in S$ where $c \neq \hat{c}$ substitute $(c,k)$ with $(\hat{c},h_{\hat{c}} + 1)$. However, by definition of $W$ such a policy cannot be budget--feasible.
\end{proof}

Under the above assumption, the optimal integer policy problem can be associated, up to a relaxation of the feasibility constraint, to the maximization of the set function $F_D(K, S)$, subject to $|S| \le W$. In the next step we show the submodularity of $F_D$. 

\begin{prope}\label{prope:submodularity}
The set function $F$ is submodular with respect to $\Omega$.\label{prope:submodularity}\vspace{-0.3cm}
\end{prope}
\begin{proof}
First, let us consider a setting with a single class. From Property~\ref{prope:threshold_optimal}, we can focus only on threshold policies and rewrite $F$ as a function of $h$, namely the threshold value (this value, in general, can be non--integer). Then it can be easily shown that $F(h)$ is a concave function since the Hessian matrix has strictly negative eigenvalues. Given a function $f: \mathbb{N} \rightarrow \mathbb{R}^+$, then $f(|S|)$ is submodular on the subsets $S$ of an arbitrary set $\Omega$ if and only if $f$ is concave. We can then conclude that $F$ is submodular in the case of a single class. Let us now show submodularity for the case with two classes. Let us denoted with $\Delta F (S|e)$ the marginal gain of $F$ obtained by adding the element $e$ to the set $S$, namely adding a transmitting slot to some class to the policy $\mu^S$. For submodularity to hold, we need to show that for every $S_a$, $S_b$, $e$ such that $S_a \subseteq S_b \subset \Omega$ and $e \in \Omega \setminus  S_b$ we have that $\Delta F (S_a|e) \ge \Delta F (S_b|e)$. By definition $e$ adds a slot to a single class, let us assume without loss of generality that this class is $c_1$. Then we have:
\begin{multline*}
\Delta F (S_a|e) = [1 - (1-(F_{c_1}(S)  +\Delta F_{c_1} (S_a|e)))(1-F_{c_2}(S))] \\ - [1 - (1-F_{c_1}(S))(1-F_{c_2}(S))] \\ =(1-F_{c_2}(S_{a}))\Delta F_{c_1} (S_a|e)
\end{multline*}
and, analogously,
$$
\Delta F (S_b|e) = (1-F_{c_2}(S_{b}))\Delta F_{c_1} (S_b|e)
$$
Since $\Delta F_{c_1} (S_a|e) \ge \Delta F_{c_1} (S_b|e)$ by submodularity of $F_{c_1}$ and  $F_{c_2}(S_{b}) \ge F_{c_2}(S_{a})$ by $F_{c_2}$ monotonicity, we have that $F$ is submodular. The same reasoning can be extended to an arbitrary number of classes.
\end{proof}

Theorem~\ref{thm:submodular_bound} can be applied by showing that Algorithm~\ref{alg:greedy} corresponds to the greedy element-selection rule reported in~(\ref{eq:greedy-selection-rule}). It is easy to see that rule~(\ref{eq:greedy-selection-rule}), when applied to the integer policy problem, proceeds by locally optimal appends in the same way that Algorithm~\ref{alg:greedy} does.  Hence, we are now in the position of state the following theorem:
\begin{thm}
Let us denote with $S^*$ the policy returned by Algorithm~\ref{alg:gridepsilon} and with $S^1_l$ is the policy constructed by Algorithm~\ref{alg:greedy} (version 1) after $l$ iterations. We then have that $F_D(K,S^1_l) \ge (1 - e^{-l/W})F_D(K,S^*)$.
\end{thm}
\begin{proof}
The inequality stated in the theorem follows immediately from the following two properties. First, by applying Theorem~\ref{thm:submodular_bound} to Algorithm~\ref{alg:greedy} (version 1) we have that $F_D(K,S^1_l) \ge (1 - e^{-l/W})F_D(K,\hat{S}^*)$. Second, since $\hat{S}^*$ is the optimal solution of a relaxed version of the integer policy problem, it holds that $F_D(K,S^*) \le F_D(K,\hat{S}^*)$ .
\end{proof}

The previous theorem, provides an \emph{online} bound on the solution quality,  being it dependent on the number of iterations the algorithm will succeed in performing without violating the budget constraint. An \emph{offline} guarantee can be given by computing the minimum number of slot $s_c$ to be assigned to each class $c$. This number can be computed by setting $\mu_{c'}(i)=1 \; \forall i \; 0 \le i \le K, c' \neq c$ and computing the maximum number of time slots during which $c$ can transmit without saturating the budget.

\begin{cor}
For any solution $\hat{S}^1$ obtained with Algorithm~\ref{alg:greedy} (version 1) we have that $F_D(K,\hat{S}^1) \ge (1 - e^{-\sum\limits_{c 	\in C}{s_c}/W})F_D(K,S^*)$.
\end{cor}

\subsubsection*{Second version, normalizing $G_1$ with budget costs}
The second version of our algorithm is an improvement to the previous version that holds when no beaconing costs are considered. Here $G_2$ is obtained by normalizing $G_1$ with the budget cost that an additional time slot will introduce. In other words, $\delta_c$ will represent a ration between benefits and costs. Under the assumption that no beaconing costs are present and that we deal with threshold policies, each transmission has an independent cost and the budget spent by a policy $S$ is given by:
$$\psi(S) =  \sum\limits_{(c,k) \in S}N_c e^{-\lambda_c \Delta (k-1)}(1 - e^{-\lambda_c \Delta})$$
and, consequently, $$G_2(\boldsymbol\mu) =\frac{G_1(\boldsymbol\mu)}{\psi(\{(c,h_c + 1)\})}$$

If we modify rule~(\ref{eq:greedy-selection-rule}) by normalizing the objective function by the budget cost for each candidate element, we can again show the equivalence between the new rule and Algorithm~\ref{alg:greedy} (version 2). As a consequence, we can again resort to a result presented in~\cite{krause12survey} and provide a quality bound on the solution obtained with the combination of the two versions of Algorithm~\ref{alg:greedy} when beaconing costs are not considered.
\begin{thm}
If no beaconing costs are present, then it holds that
$$\max \{ F_D(K,\hat{S}^1), F_D(K,\hat{S}^2) \} \ge \frac{1}{2} (1 - \frac{1}{e}) \max\limits_{S \subseteq \Omega \\ \psi(S) \le \Psi} F_D(K, S)$$\vspace{-0.3cm}
\end{thm}
\begin{proof}
The proof follows immediately by the consideration made above and a straightforward adaptation of results presented in~\cite{krause12survey}.
\end{proof}

%
%


\section{Experimental evaluation}
In this section, we provide some experimental evaluations of the proposed algorithms. Results are obtained from MATLAB simulations and are aimed at showing the feasibility of our approach and evaluating its performance in terms of solution's quality. We shall also discuss some qualitative issues observed in the obtained policies.

\subsection{Experimental setting}
Each instance of our problem is described by different parameters. In our experiments, we generated istances by considering finite sets of values for each parameter, see Table~\ref{tab:parameters} for a complete summary. In particular, we devote our attention to three different mobility profiles and to three different transmission technologies for mobile nodes. Mobility profiles are characterized by increasing average speeds. The scenario that we imagine is populated by mobile devices carried by pedestrians, users on bycycles, and users on vehicles, respectively. The transmission technologies we consider provide increasing communication ranges: ZigBee, Bluetooth 4.0. and Wi-Fi Direct. We derive the corresponding values for $\rho$ and $\beta$ by considering the technical specifications of each technology and assuming an application scenario where a single packet has a size of $5$kB and a slot interval $\Delta=10s$. For simplicity, we assign the same number of users to each class.

\begin{table}[h]
\caption{Parameters used for experiments}
\label{tab:parameters}
\centering
\begin{tabular}{|c|}
\hline
temporal deadline for delivery ($\tau$)\\ 
\hline
25, 50, 100, 250 \\
\hline
\end{tabular}
\vspace{0.1cm}
\begin{tabular}{|c|}
\hline
radius of the environment ($L$)\\ 
\hline
350, 500, 750,1000 \\
\hline
\end{tabular}
\begin{tabular}{|c|}
\hline
number of nodes $N_c$\\ 
\hline
9, 15, 20 \\
\hline
\end{tabular}
\begin{tabular}{|c|}
\hline
mobility profiles ($v_c$) \\
\hline
pedestrians (1.5$m/s$)\\
bicycles (6$m/s$)\\
vehicles (9$m/s$)\\
\hline
\end{tabular}
\begin{tabular}{|c|}
\hline
transmission technologies \\
\hline
ZigBee ($R=15m$) \\
Bluetooth 4.0 ($R=50m$)\\
Wi-Fi Direct ($R=100m$)\\
\hline
\end{tabular}\vspace{-0.5cm}
\end{table}

In the experimental results proposed here, we consider up to 3 classes and a temporal discretization varying according to $\epsilon \in \{1, 1/3, 1/5\}$. The reason behind this choice can be intuitively described by the two graphs of Figure~\ref{fig:bounds}, where we depict the theoretical lower bound from Theorem~\ref{thm:lower_bound} with respect to different resolutions and numbers of classes. As it can be seen, a maximum resolution of $\epsilon = 1/5$ represents a reasonable choice to guarantee about $95\%$ of the optimal solution quality without the burden of a prohibitive number of time slots. On the other side, by adopting a maximum number of 3 classes we obtain a case which is fairly close to the worst case (derived for an infinite number of classes) and that is computable by means of our grid algorithm (as discussed in the following---we recall that our grid search requires compute time that is exponential in the number of classes). Finally, we remark that we chose a small number of nodes for simplicity in our experiments, but that the compute time of all our algorithms is linear in the number of nodes.

\begin{figure}[!htbp]
\vspace{-0.2cm}
\centering
\includegraphics[scale=0.35]{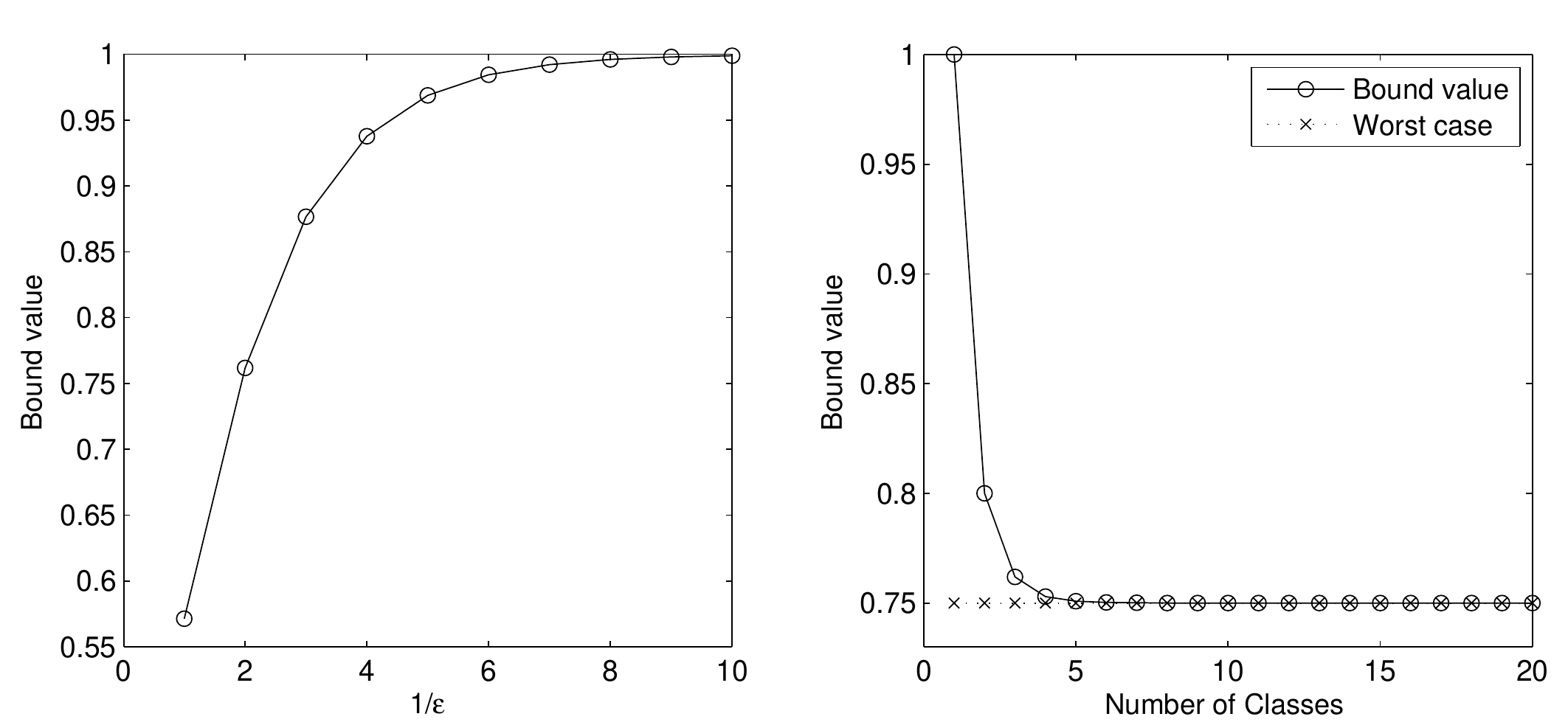}
\vspace{-0.2cm}
\caption{Theoretical lower bound over solution's quality (Theorem~\ref{thm:lower_bound}).}
\vspace{-0.5cm}
\label{fig:bounds}
\end{figure}

\subsection{Benchmarks}

We compare the performance of our algorithm w.r.t. the performance of two heuristic easily--computable algorithms and of an upper bound over the value of the optimal policy.

\subsubsection{Greedy on arrival rate} 

This algorithm works as follows: it sorts the classes in descending order of $\lambda_c$, then it allocates all the possible budget to the classes from the first one in the order to the last one. For instance, given three classes with $\lambda_1=0.3,\lambda_2=0.2,\lambda_3=0.1$, the algorithm assigns all the possible budget to class $1$ and, if there is a remaining budget, then all the remaining budget is assigned to class $2$ and so on. The rationale is that we expect that the larger the arrival rate the larger the delivery probability. The complexity of this algorithm is obviously easy given that the policy can be found by solving at most $|C|$ equations.

\subsubsection{Class--independent policies}

This algorithm searches for the optimal solution of an overconstrained problem in which: the policies related to all the classes are the same, formally $\mu(k)=\mu_c(k)$ for all $c$, and, when the policy is probabilistic, then either the source transmits to all the classes or it does not transmit at all. This last assumption leads to a new formulation of the budget constraint:
\[
\sum_{c\in C} \rho_c N_c \cdot (1- Q_{c,0,K}(\mu)) + \sum_{\omega\in \Omega}\sum_{k=0}^{K-1}\beta_\omega\cdot \mu(k) \leq \Psi			
\]
By Property~\ref{property:budget}, the optimal policy is such that the budget $\Psi$ is completely consumed and therefore the above inequality holds with equality. Therefore, the optimization problem reduces to the problem of finding the policy that completely consumes the budget. Formally, interpreting the (class--independent) threshold $h$ as a continuous variable, we can write:
\[
g(h)= \sum_{c\in C} N_c \cdot e^{-\lambda_c \Delta h} - \sum_{\omega\in \Omega}\left(\sum_{c\in C_\omega} N_c -\frac{\beta_\omega}{\rho_c} \cdot h+ \frac{\Psi}{\rho_c} \right)= 0
\]
Function $g$ is a single--variable function strictly monotonically decreasing in $h$ and infinitely differentiable. Such a function admits only one zero, and therefore the above equation admits only one solution. Such a solution can be found (approximately) by using the Newton method, that in this case, due to the property of the function, has a quadratic convergence speed (the number of correct digits roughly at least doubles in every step). Thus, we obtain an approximate solution of high quality within very short time.

\subsubsection{Upper bound over the optimal value}

An over bound over the value of the optimal solution can be found by using a variation of the algorithm described in Section~\ref{subsection:gridalgorithm}. More precisely, we use Algorithm~\ref{alg:gridepsilon} to enumerate all the policies consuming entirely the budget and we change each policy rounding each $h_c$ to the smallest integer and then adding $1$ for every $c$. Notice that these new policies violate the budget constraint. Among all these policies we find the one maximizing the delivery probability. Its value is an upper bound over the value of the optimal policy. In the graphs we denote this value as $UB$. The proof follows. Call $\boldsymbol \mu^*$ the optimal policy profile with (potentially fractional) thresholds $h_c^*$. Call $\hat{\boldsymbol \mu}$ a generic policy profile obtained as described above. It can be easily  observed (it follows from the fact that, fixed the policies of all the classes but one, the policy of the remaining class that consumes entirely the budget is always one) that there alway exists a policy profile $\hat{\boldsymbol \mu}$ such that $\hat{h}_c\geq h^*_c$ for all $c$. Therefore, given that the objective function is strictly monotone in $h_c$, the objective value of  $\hat{\boldsymbol \mu}$ is strictly better than the value $\boldsymbol \mu^*$.

\subsection{Experimental results}
Figure~\ref{fig:mean_fub_tau} reports how $F_D / UB$ varies as the values of the parameters $\tau,L,N_c$ vary as summarized in~Table~\ref{tab:parameters}, $|C|\in\{1,2,3\}$, and $\frac{1}{\epsilon}=5$. For each parameter, we average $F_D / UB$ over the other instances sharing the same value for that parameter. It can be observed that grid search and greedy constructions obtain a remarkable better performance in each case when compared with the benchmarking greedy algorithms based on the arrival rate and the class-independent one. Not exploiting the knowledge about the different classes and solely considering the arrival rate turned out to achieve very similar performances. By increasing the value of $\tau$, it can be seen how this gap with the benchmarks shrinks, suggesting the intuition that when the deadline for packet delivery is large even simplistic policies are able to obtain good delivery probabilities. Another aspect that can be observed is that greedy constructions revealed to be quite effective for the tested cases, since they were able to obtain high performances comparable to the grid search. By increasing the value of $L$, it can be seen how this gap with the benchmarks increases, instead the gap keeps to be approximately constant as $N_c$ and $C$ vary. Interestingly, the approximation ratio of our algorithms is  constant (i.e., $>99\%$) w.r.t. all the parameters values.
\begin{figure}[!htbp]
\centering
\vspace{-0.2cm}
\includegraphics[scale=0.4]{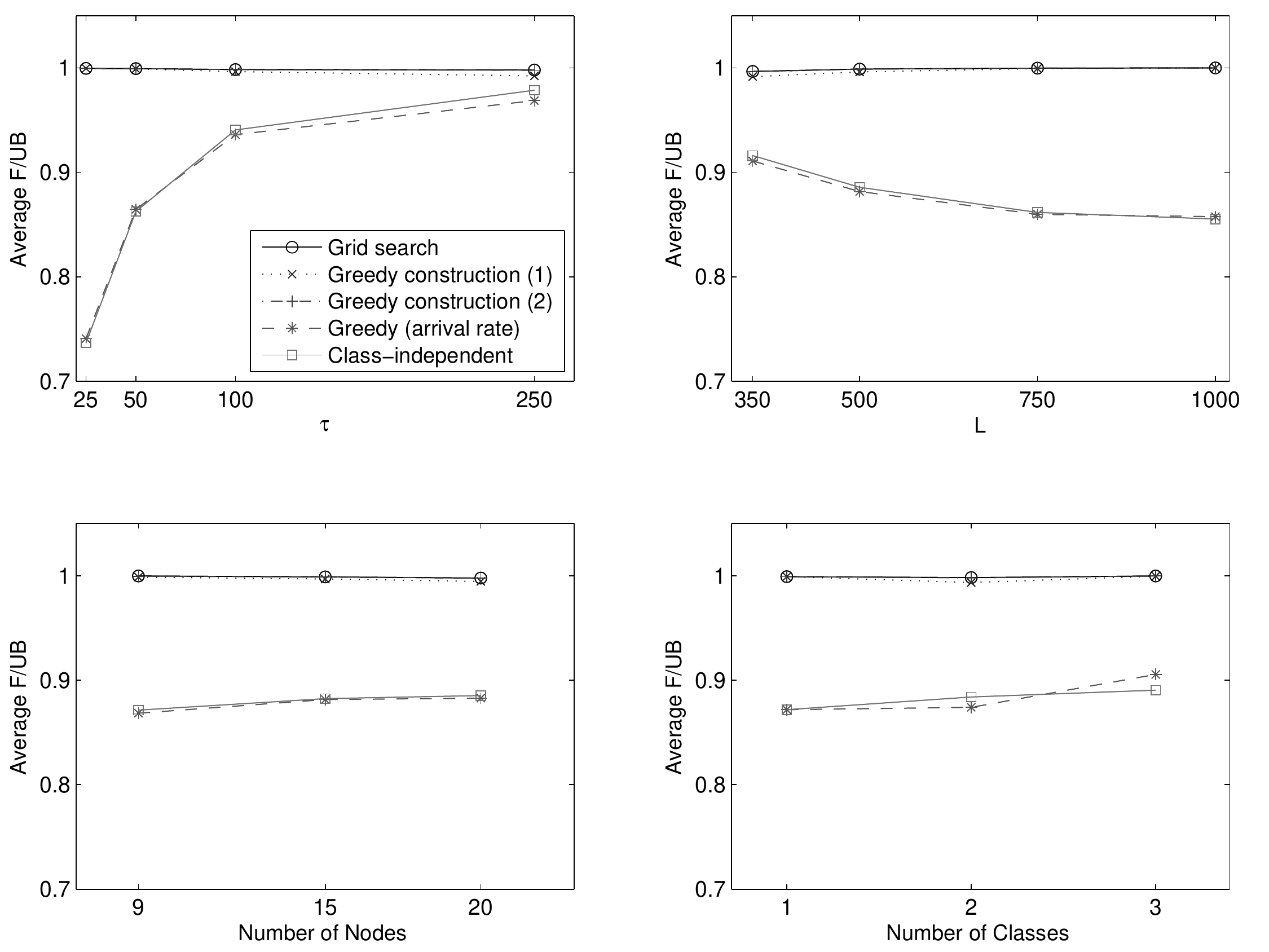}
\vspace{-0.4cm}
\caption{Average $F_D/UB$ with respect to different parameters at $\frac{1}{\epsilon}=5$.}
\vspace{-0.4cm}
\label{fig:mean_fub_tau}
\end{figure}

A more detailed overview on how the performance varies with respect to $\tau$ is shown by the boxplots of Figure~\ref{fig:fub_tau_boxplot}. These graphs show the similarity in performance between the grid search and the greedy constructions algorithms. These last ones obtained worse performances for a limited number of outlier instances. Also it is evident how having finer resolutions remarkably improves the solution's quality.
\begin{figure}[!htbp]
\centering
\vspace{-0.2cm}
\includegraphics[scale=0.45]{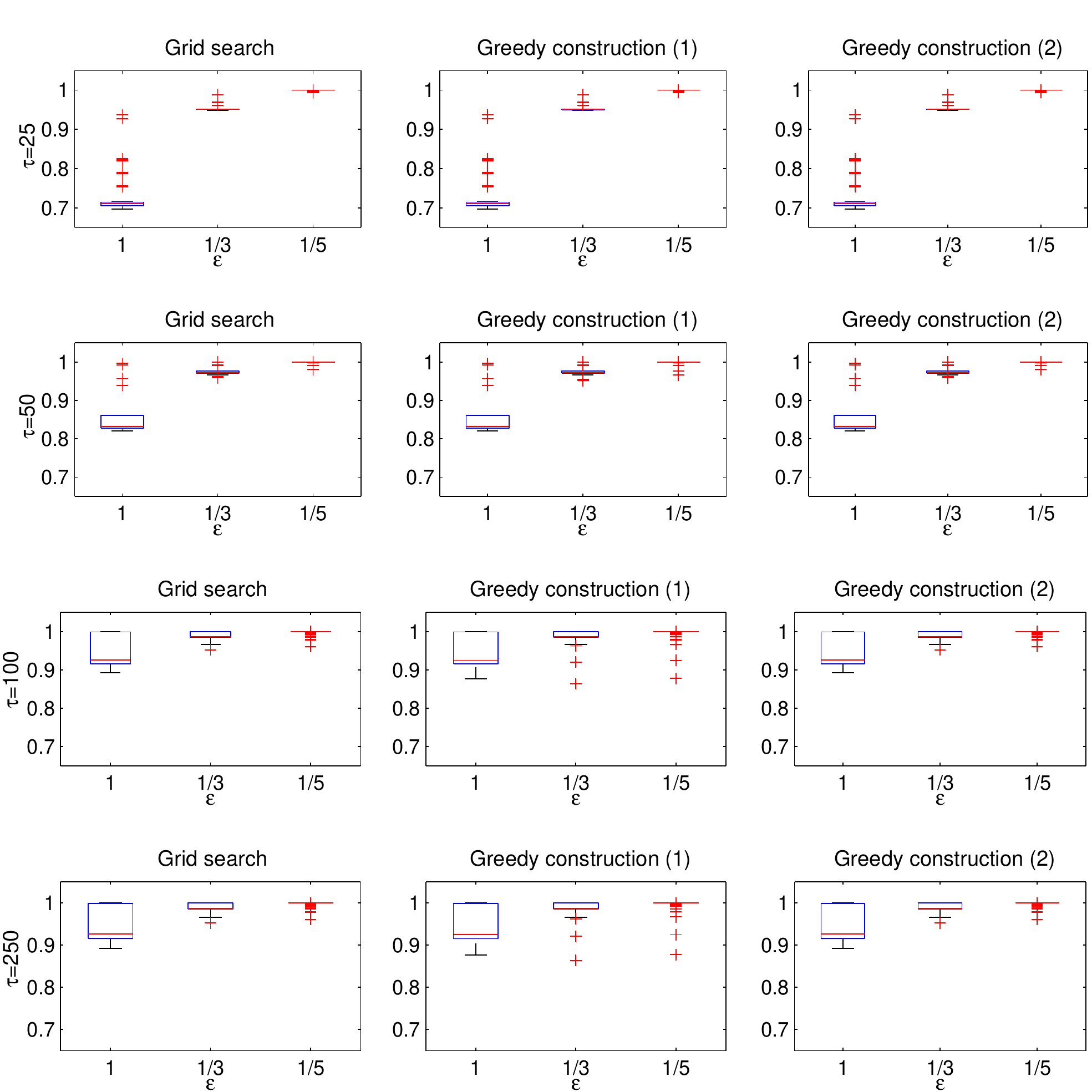}
\vspace{-0.4cm}
\caption{Boxplots showing $F_D/UB$ w.r.t. $\tau$ for different algorithms.}
\vspace{-0.4cm}
\label{fig:fub_tau_boxplot}
\end{figure}

The above results suggest that greedy constructions seem to be quite effective approaches to approximate the optimal policy requiring, at the same time, much less computational effort than the grid search. In Figure~\ref{fig:time_greedy}, we show a comparison between computational times obtained with the grid search and the greedy construction algorithms respectively. In particular, we evaluated the algorithms' scalability when the number of classes grows. To obtain these results we fixed the values of some parameters ($\epsilon = 1/3$, $\tau = 100$, $N_c=10$, $L=500$) and we generated random mobility profiles and transmission technologies by uniformly sampling from the following intervals: $R_c \in [15,50]$, , $v_c \in [1, 15]$ $\rho_i \in [0.05, 0.25]$ $\beta_c \in [3 \times 10^{-7}, 8 \times 10^{-7}]$. It is easy to see how grid search shows an exponential growth in time, while greedy construction proved to be much more efficient even for larger number of classes. Considering a deadline of 1 hour, grid search was not able to compute a solution for more than 4 classes, while greedy construction managed to compute solution up to $800$ classes.
\begin{figure}[!htbp]
\centering
\vspace{-0.3cm}
\includegraphics[scale=0.4]{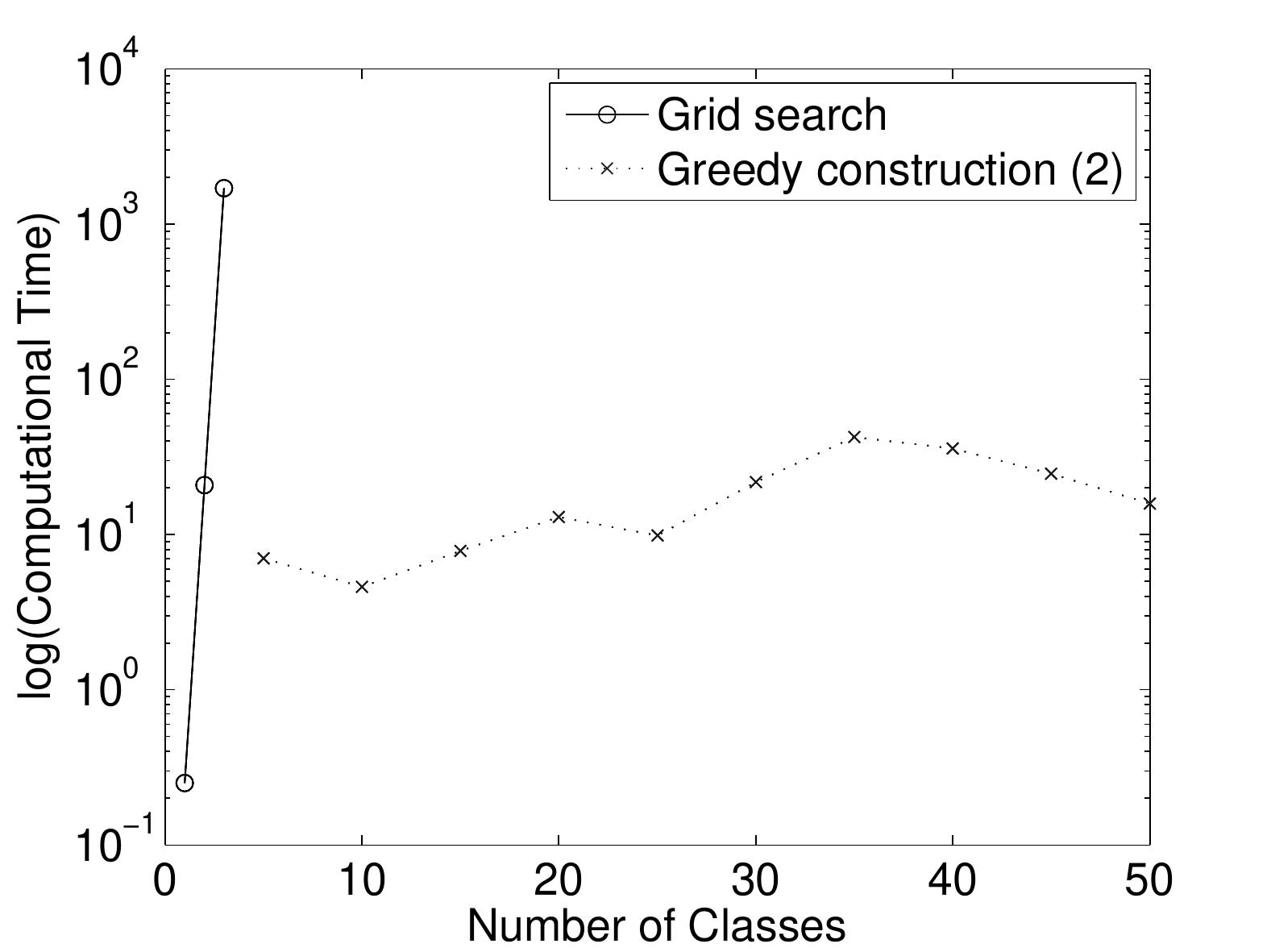}\vspace{-0.2cm}
\caption{Time (in seconds) scalability with the number of classes.}
\label{fig:time_greedy}
\vspace{-0.5cm}
\end{figure}

Finally, Figure~\ref{fig:policies} depicts a qualitative evaluation of the policies returned by our algorithms. We consider a reference value for the budget upper bound $\Psi$ and we show how the thresholds of the optimal policy (obtained with grid search) are distributed over the three different technologies. It can be observed how, by increasing the budget, the optimal policy tends to schedule transmissions with all the three technologies. When the budget gets smaller and smaller, then the policy tries to rely more on those technology that have a longer communication range.

\begin{figure}[!htbp]
\centering
\vspace{-0.2cm}
\includegraphics[scale=0.4]{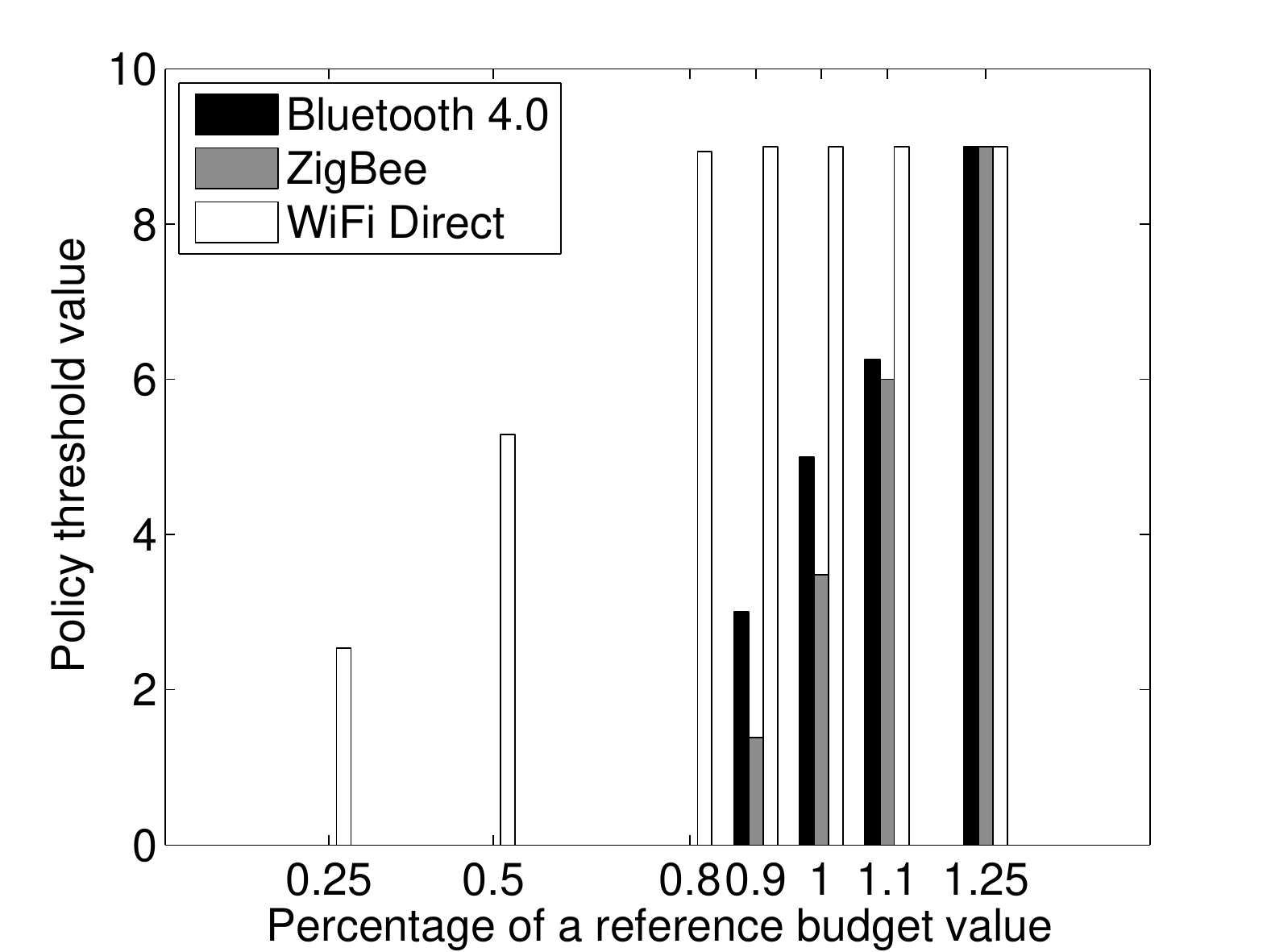}
\vspace{-0.2cm}
\caption{Policy thresholds with different upper bounds on budget.}
\label{fig:policies}
\vspace{-0.5cm}
\end{figure}

\section{Conclusions}

In this paper we studied two-hop routing for Delay Tolerant Networks when heterogeneous technologies are present taking into account beaconing signal and deadlines after which nodes discard packets. Differently from the literature, whose adopts fluid approximation to find optimal policies---providing exact solution in the limit when the number of nodes is infinite, but a coarse approximation otherwise---, we adopt an operations research approach, formulating the problem as an optimization problem and designing approximation schemes with theoretical bounds. We thoroughly evaluated our algorithms with realistic settings in terms of approximation ratio and compute time as the parameters change. We experimentally showed that for all the generated instances our algorithms have an approximation ratio larger than 99\% and that they scale linearly as the values of the parameters increase and therefore they can be applied with extremely large instances. 

\vspace{-0.5cm}


\bibliographystyle{plain}
\bibliography{citations}

\begin{thebibliography}{10}

\bibitem{DBLP:journals/ton/AltmanABP13}
E.~Altman, A.~Prakash Azad, T.~Basar, and F.~De Pellegrini.
\newblock Combined optimal control of activation and transmission in
  delay-tolerant networks.
\newblock {\em IEEE ACM T NETWORK}, 21(2):482--494, 2013.

\bibitem{DBLP:journals/tac/AltmanBP11}
E.~Altman, T.~Basar, and F.~De Pellegrini.
\newblock Optimal control in two--hop relay routing.
\newblock {\em IEEE T AUTOMAT CONTR}, 56(3):670--675, 2011.

\bibitem{DBLP:conf/infocom/AltmanNPM09}
E.~Altman, G.~Neglia, F.~De Pellegrini, and D.~Miorandi.
\newblock Decentralized stochastic control of delay tolerant networks.
\newblock In {\em INFOCOM}, pages 1134--1142, 2009.

\bibitem{DBLP:journals/ton/AltmanP11}
E.~Altman and F.~De Pellegrini.
\newblock Forward correction and fountain codes in delay--tolerant networks.
\newblock {\em IEEE ACM T NETWORK}, 19(1):1--13, 2011.

\bibitem{DBLP:journals/twc/AltmanSP13}
E.~Altman, L.~Sassatelli, and F.~De Pellegrini.
\newblock Dynamic control of coding for progressive packet arrivals in {DTN}s.
\newblock {\em IEEE T WIREL COMMUN}, 12(2):725--735, 2013.

\bibitem{DBLP:journals/cn/AzouziPSK13}
R.~E. Azouzi, F.~De Pellegrini, H.~B.~A. Sidi, and V.~Kamble.
\newblock Evolutionary forwarding games in delay tolerant networks: Equilibria,
  mechanism design and stochastic approximation.
\newblock {\em COMPUT NETW}, 57(4):1003--1018, 2013.

\bibitem{DBLP:journals/ton/BanerjeeCL10}
N.~Banerjee, M.~D. Corner, and B.~Neil Levine.
\newblock Design and field experimentation of an energy--efficient architecture
  for dtn throwboxes.
\newblock {\em IEEE ACM T NETWORK}, 18(2):554--567, 2010.

\bibitem{DBLP:conf/wd/ChahinAPA11}
W.~Chahin, R.~El Azouzi, F.~De Pellegrini, and A.~P. Azad.
\newblock Blind online optimal forwarding in heterogeneous delay tolerant
  networks.
\newblock In {\em Wireless Days}, pages 1--6, 2011.

\bibitem{chen2010mobicent}
B.B. Chen and M.C. Chan.
\newblock Mobicent: a credit-based incentive system for disruption tolerant
  network.
\newblock In {\em INFOCOM}, pages 1--9, 2010.

\bibitem{el2010evolutionary}
R.~El-Azouzi, F.~De~Pellegrini, and V.~Kamble.
\newblock Evolutionary forwarding games in delay tolerant networks.
\newblock In {\em WiOpt}, pages 76--84, 2010.

\bibitem{DBLP:journals/pe/GroeneveltNK05}
R.~Groenevelt, P.~Nain, and G.~Koole.
\newblock The message delay in mobile ad hoc networks.
\newblock {\em PERFORM EVALUATION}, 62(1--4):210--228, 2005.

\bibitem{DBLP:journals/comcom/GuerrieriCPMM10}
A.~Guerrieri, I.~Carreras, F.~De Pellegrini, D.~Miorandi, and A.~Montresor.
\newblock Distributed estimation of global parameters in delay--tolerant
  networks.
\newblock {\em COMPUT COMMUN}, 33(13):1472--1482, 2010.

\bibitem{jain2004routing}
Sushant Jain, Kevin Fall, and Rabin Patra.
\newblock Routing in a delay tolerant network.
\newblock In {\em ACM SIGCOMM '04}, pages 145--158, 2004.

\bibitem{krause12survey}
A.~Krause and D.~Golovin.
\newblock Submodular function maximization.
\newblock In {\em Tractability: Practical Approaches to Hard Problems (to
  appear)}. Cambridge University Press, 2012.

\bibitem{li2010routing}
Q.~Li, S.~Zhu, and G.~Cao.
\newblock Routing in socially selfish delay tolerant networks.
\newblock In {\em IEEE INFOCOM}, pages 1--9, 2010.

\bibitem{ning2011incentive}
T.~Ning, Z.~Yang, X.~Xie, and H.~Wu.
\newblock Incentive-aware data dissemination in delay-tolerant mobile networks.
\newblock SECON, 2011.

\bibitem{niyato2010coalition}
D.~Niyato, P.~Wang, W.~Saad, and A.~Hjorungnes.
\newblock Coalition formation games for improving data delivery in delay
  tolerant networks.
\newblock In {\em GLOBECOM}, pages 1--5, 2010.

\bibitem{DBLP:conf/wiopt/PellegriniAB10}
F.~De Pellegrini, E.~Altman, and T.~Basar.
\newblock Optimal monotone forwarding policies in delay tolerant mobile ad hoc
  networks with multiple classes of nodes.
\newblock In {\em WiOpt}, pages 497--504, 2010.

\bibitem{DBLP:journals/ior/PerryW11}
O.~Perry and W.~Whitt.
\newblock A fluid approximation for service systems responding to unexpected
  overloads.
\newblock {\em OPER RES}, 59(5):1159--1170, 2011.

\bibitem{DBLP:journals/orl/Randhawa13}
R.~S. Randhawa.
\newblock Accuracy of fluid approximations for queueing systems with
  congestion-sensitive demand and implications for capacity sizing.
\newblock {\em OPER RES LETT}, 41(1):27--31, 2013.

\bibitem{DBLP:journals/tsmc/TembineAAH10}
H.~Tembine, E.~Altman, R.~El Azouzi, and Y.~Hayel.
\newblock Evolutionary games in wireless networks.
\newblock {\em IEEE T SYST MAN CY B}, 40(3):634--646, 2010.

\bibitem{zhu2009smart}
H.~Zhu, X.~Lin, R.~Lu, Y.~Fan, and X.~Shen.
\newblock Smart: A secure multilayer credit-based incentive scheme for
  delay-tolerant networks.
\newblock {\em IEEE T VEH TECHNOL}, 58(8):4628--4639, 2009.

\end{thebibliography}

\begin{IEEEbiography}{Michael Shell}
Biography text here.
\end{IEEEbiography}

\begin{IEEEbiographynophoto}{John Smith}
Biography text here.
\end{IEEEbiographynophoto}


\begin{IEEEbiographynophoto}{Jane Smith}
Biography text here.
\end{IEEEbiographynophoto}

\end{document}